\documentclass[11pt]{article}
\usepackage{amsmath}
\usepackage{amsthm}
\usepackage{amsfonts}
\usepackage{amssymb}
\usepackage{geometry}
\usepackage[matrix,arrow,curve]{xy}

\newcommand{\Fq}{\mathbb{F}_q}
\newcommand{\Fm}{\mathbb{F}_{q^m}}
\newcommand{\Fn}{\mathbb{F}_{q^n}}
\newcommand{\F}{\mathbb{F}}
\newcommand{\C}{\mathcal C}
\DeclareMathOperator*{\rk}{rk}
\newcommand{\pr}{\mathrm{Pr}}

\newtheorem{assumption}{Assumption}
\newtheorem{definition}{Definition}
\newtheorem{theorem}{Theorem}
\newtheorem{proposition}{Proposition}
\newtheorem{lemma}{Lemma}

\begin{document}

\title{Fuzzy Authentication using Rank Distance}
% Use \titlerunning{Short Title} for an abbreviated version of
% your contribution title if the original one is too long
\author{Alessandro Neri, Joachim Rosenthal and Davide Schipani\thanks{Institute of Mathematics, University of Zurich, Switzerland}}
% Use \authorrunning{Short Title} for an abbreviated version of
% your contribution title if the original one is too long
%\institute{Alessandro Neri \at Institute of Mathematics, University of Zurich, Winterthurerstr. 190, 8057 Zurich, Switzerland, %\email{alessandro.neri@math.uzh.ch}
%\and Joachim Rosenthal \at Institute of Mathematics, University of Zurich, Winterthurerstr. 190, 8057 Zurich, Switzerland, %\email{rosenthal@math.uzh.ch}
%\and Davide Schipani \at Institute of Mathematics, University of Zurich, Winterthurerstr. 190, 8057 Zurich, Switzerland, %\email{davide.schipani@math.uzh.ch}}

%
% Use the package "url.sty" to avoid
% problems with special characters
% used in your e-mail or web address
%
\maketitle

\abstract{
Fuzzy authentication allows authentication based on the fuzzy matching of two objects, for example based on the similarity of two strings in the Hamming metric, or on the similiarity of two sets in the set difference metric. Aim of this paper is to show other models and algorithms of secure fuzzy authentication, which can be performed using the rank metric. A few schemes are presented which can then be applied in different scenarios and applications. 
	\vspace{1cm}
}

\section{Introduction}
\label{sec:1}

Recent years have seen a lot of research around the problem of authentication using approximate matching under a certain metric of similarity, while still enabling a secure storage of sensible authentication data. The typical, but not the only scenario, where such a system is needed, is in the use of biometric features, like fingerprints, for authentication purposes.

Several models have been proposed that may be more appropriate for different applications. For example the fuzzy commitment scheme \cite{ju99} models data as bit strings and compares strings in the Hamming metric; the fuzzy vault \cite{ju02} models data as sets of elements and compares sets in the set difference metric.

In this paper we present fuzzy authentication schemes using the rank metric by generalizing the schemes mentioned above for other model scenarios and highlighting possible applications. The structure of the paper is the following. Section \ref{sec:3} recalls some mathematical concepts and definitions concerning rank metric codes and linearized polynomials. Section \ref{sec:2} presents the fuzzy commitment scheme in the rank distance, a model whereby the tolerance needed in the authentication is not based on the number of different bits between two strings but on the similarity of two matrices, more precisely on the rank of their difference. Section \ref{sec:4} is devoted to a fuzzy vault scheme using linearized polynomials, which relates the set difference with the rank metric. The scheme is an alternative to the standard fuzzy vault based on Reed-Solomon decoding.
Section \ref{sec:5} gives hints on possible applications and model scenarios of the schemes presented in the previous sections.

\section{Rank Metric Codes and Linearized Polynomials}
\label{sec:3}

Let $q$ be a prime power and let $\Fq$ denote the finite field with $q$ elements.
Recall that 
 $\Fm$ is isomorphic (as
a vector space over $\Fq$) to the vector space $\Fq^m$.
% If not noted differently we will use the isomorphism
% \begin{align*}
%   \F_q^m &\longrightarrow \F_{q^m}\cong \F_q[\alpha] \\
%   (v_1, \dots, v_m) &\longmapsto \sum_{i=1}^m v_i \alpha^{i-1} .
% \end{align*}
One then easily obtains the isomorphic description of matrices over
the base field $\F_q$ as vectors over the extension field, i.e.\
$\Fq^{m\times n}\cong \Fm^n$.

\begin{definition}
  The \emph{rank distance} $d_R$ on $\Fq^{m\times n}$ is defined by
  \[d_R(X,Y):= \rk(X-Y) , \quad X,Y \in \Fq^{m\times n}. \]
  In the same way it is possible to define the rank distance between two elements
  $\boldsymbol x,\boldsymbol y \in \Fm^n$ as the rank of the
  difference of the respective matrix representations in
  $\Fq^{m\times n}$.
\end{definition}

A \emph{rank metric code} $\C$ is a subset of $\Fq^{m\times n}$ (or $\Fm^n$) equipped with the rank distance. The \emph{minimum distance} of a rank metric code $\C$ is the quantity
$$d_R(\C):=\min\left\{d_R(u,v) \mid u,v\in \C, u\neq v\right\}. $$
We can define special classes of rank metric codes introducing linearity. 
 An $\Fm$-linear rank metric code of dimension $k$ is a rank metric code that is also a $k$-dimensional subspace of the $\Fm$-vector space $\Fm^n$. 
An $\Fq$-linear rank metric code of dimension $k'$ is a rank metric code that is also a $k'$-dimensional subspace of the $\Fq$-vector space
$\Fm^n\cong \Fq^{m\times n}$.

Observe that an $\Fm$-linear rank metric code of dimension $k$ is also an $\Fq$-linear code of dimension $mk$.

We will use the notation $[n,k,d]$-code  for a $k$-dimensional $\Fm$-linear code with minimum distance $d$, and $[nm,k',d']$-code  for a $k'$-dimensional $\Fq$-linear code with minimum distance $d'$.

\begin{theorem}[Singleton-like Bound]\label{th:SB}
Let $\C\subseteq \Fq^{m\times n}$ be a rank metric code. Then
$$|\C|\leq \min\left\{q^{m(n-d+1)},q^{n(m-d+1)}\right\}. $$
\end{theorem}

\begin{proof}
See \cite{ga85}, or \cite[Theorem 1]{ro91}.
\end{proof}

\begin{definition}
Codes attaining the Singleton-like bound are called \emph{Maximum Rank Distance (MRD) Codes}.
\end{definition}

When $n\leq m$ a class of codes attaining the Singleton-like bound was first proposed in \cite{ga85} and then generalized in \cite{ks05}. These codes are $\Fm$-linear rank metric codes.
Let $(v_1,\dots, v_n) \in \Fm^n$ be a vector, we denote the $k
\times n$ \emph{$s$-Moore matrix} by
\[M_{s,k}(v_1,\dots, v_n) := \left( \begin{array}{cccc} v_1 & v_2
    &\dots &v_n \\ v_1^{[s]} & v_2^{[s]} &\dots &v_n^{[s]} \\
    \vdots&&&\vdots \\ v_1^{[s(k-1)]} & v_2^{[s(k-1)]} &\dots
    &v_n^{[s(k-1)]} \end{array}\right) ,\] where $[i]:= q^i$.
\begin{definition}\label{def:Gab}
  Let $g_1,\dots, g_n \in \Fm$ be linearly independent over
  $\F_q$ and let $s$ be coprime to $m$. We define a \emph{generalized
    Gabidulin code} $\mathcal{C}\subseteq \Fm^{n}$ of dimension
  $k$ as the linear block code with generator matrix
  $M_{s,k}(g_1,\dots, g_n)$.  Using the isomorphic matrix
  representation we can interpret $\mathcal{C}$ as a matrix code in
  $\Fq^{m\times n}$.
\end{definition}

These codes are optimum for rank distance, since they are $[n,k,n-k+1]$-codes. Moreover, for this class of codes there exist
polynomial-time decoding algorithms decoding up to their
error-correcting capability $t=\left\lfloor\frac{n-k}{2}\right\rfloor$,  \cite{lo06, ri04, si09}.

Observe that when $s=1$, this definition of Gabidulin codes is the $q$-analog of Reed-Solomon codes with the Hamming distance. Here, a set of distinct elements is replaced by a set of linearly independent elements, and the power $g_i^j$ is replaced by the Frobenius power $g_i^{[j]}$. Reed-Solomon codes can also be seen as evaluation of polynomials of degree less than $k$ in $n$ distinct
points. We can give a $q$-analog of this interpretation for Gabidulin codes, as evaluation of linearized polynomials in $n$ linearly independent elements.

\begin{definition}
A \emph{linearized polynomial} over $\Fm$ is a polynomial $f(x)\in \Fm[x]/(x^{q^m}-x)$ of the form
$$\sum_{i=0}^{m-1}f_ix^{[i]}.$$
We denote by $\mathcal L_m(\Fm)$ the space of linearized polynomials over $\Fm$.
\end{definition}

Let $\mathcal G_{k,s} \subseteq \mathcal L_m(\Fm)$ be the set defined as
$$ \mathcal G_{k,s}:= \left\{f_0x+f_1x^{[s]}+\ldots+f_{k-1}x^{[s(k-1)]} \mid f_i\in \Fm\right\}.$$

\begin{proposition}
Let  $g_1,\dots, g_n \in \Fm$ be linearly independent over
  $\F_q$ and let $s$ be an integer coprime to $m$. Let moreover $\C$ be the Generalized Gabidulin code whose generator matrix is  $M_{s,k}(g_1,\dots, g_n)$. Then 
$$\C= \left\{(f(g_1),f(g_2),\ldots,f(g_n))\mid f\in \mathcal G_{k,s}\right\}.$$
\end{proposition}
From now on we will write $\mathcal G_{k,s}(g_1,\ldots, g_n)$ for such a code.

For many years Gabidulin codes have been the only known MRD codes over $\Fm$. Recently some construction of non-Gabidulin MRD 
codes have been discovered \cite{co16, cr15}, but  many of these codes are not linear over $\Fm$. Some constructions of linear
non-Gabidulin MRD codes can be found in \cite{ho16} and as a special class of the codes presented in \cite{sh15}.

Although there are few known constructions of MRD codes, it was shown in \cite{ne16} that most linear rank metric codes are MRD and that the Gabidulin codes are only a small franction among the  MRD codes.

\section{Fuzzy Commitment Scheme with the Rank Distance}
\label{sec:2}

In 1999 Juels and Wattenberg \cite{ju99} proposed a fuzzy commitment scheme to allow fuzzy authentication with secure storage of biometric data in binary form. In \cite{sc10} the authors revisited the scheme in the setting of an arbitrary finite field by focusing on implementations and security concerns. In \cite{sc11} they proposed a dual version of the scheme, called fuzzy syndrome hashing, featuring some advantages in terms of security and use of iterative decoding. In \cite{sc12} they presented scenarios involving burst error correction and higher dimensional data.

Here we are going to describe a new fuzzy commitment scheme using the rank metric. 
In a following section about applications, we will describe a few scenarios where this scheme can be applied.

In our authentication model, we wish to consider two vectors $b,b' \in \Fm^n$  (or, equivalently, their matrix representations $B, B'\in \Fq^{m\times n}$) as belonging to the same person or entity as long as their rank distance is less than a certain predetermined threshold. And for security concerns we do not want to store vectors (or matrices) unencrypted.

Suppose now that we have a rank metric code $\C\subseteq \Fm^n$ whose minimum distance is $d=2t+1$ and assume there exists an efficient algorithm for decoding up to $t$ errors.

Let $h:\Fq^{m\times n} \rightarrow \Fq^{m\times r}$ be a collision resistant hash function, i.e. such that it is not feasible to compute an $u\in h^{-1}(v)$ for any $v \in \Fq^{m\times r}$. Observe that a hash function $h':\Fm^{n} \rightarrow \Fm^{r}$ can be defined starting from $h$, as the diagram

$$\xymatrix{\Fq^{m\times n}  \ar[r]^-{h} \ar[d] & \Fq^{m\times r} \ar[d] \\
\Fm^n \ar@{.>}[r]^{h'} \ar[u] & \Fm^r\ar[u]\\
}$$
shows.

%%representation $B\in \Fq^{m\times n}$. At this point 
As in the standard fuzzy commitment scheme, we select at random a codeword $c_b \in \C$ and we store the tuple
$$(l, h(c_b))$$
where $l=b-c_b$. 

This scheme is essentially %the same proposed in \cite{sc10} 
the analogue of the standard fuzzy commitment with the difference that we use rank metric codes instead of Hamming codes.
Analogously as in \cite{sc10} one can show the following result.

\begin{theorem}%\cite[Theorem 1]{sc10}
  If  $b \in \Fm^n$ can be chosen uniformly over the entire ambient space $\Fm^n$, then computing $ b \in \Fm^n$ from the stored data $(l,h(c_b))$ is computationally equivalent to invert the ‘restricted’ hash function 
 $$h_{|_\C}: \C \longrightarrow \Fm^r.$$ 
\end{theorem}

\section{A Linearized Polynomial Fuzzy Vault Scheme}
\label{sec:4}
The polynomial fuzzy vault (PFV) scheme was introduced in \cite{ju02} and allows fuzzy authentication in the set-difference metric.
In \cite{ma16} the authors proposed
a fuzzy vault scheme using codes in another metric, relating the set difference with the subspace distance on the set of Grassmanians.
The PFV scheme can also be generalized in a natural way using linearized polynomials and codes over the rank metric as follows.

First, we make the following assumption about the set of features used for authentication, both the set initially used to build the vault and the one submitted later for authentication.
\begin{assumption}\label{as:lin}
 Assume that the set of features ($A$ or $W$ in the following) is given by $n$ $\Fq$-linearly independent elements in $\Fn$, i.e. it is an $\Fq$-basis for $\Fn$.
\end{assumption}

This is usually not a restrictive assumption given the follwing result:
\begin{lemma}\label{lem:lin}
If the features %elements of $W$ 
are chosen with uniform distribution, then Assumption \ref{as:lin} is satisfied with probability 
$$ \prod_{i=0}^{n-1}\frac{(q^n-q^i)}{(q^n-i)}\geq \prod_{i=0}^{n-1}(1-q^{i-n}).$$
\end{lemma}

\begin{proof}
 The number of $\Fq$-basis of $\Fn$ is $\frac{\prod_{i=0}^{n-1}(q^n-q^i)}{n!}$, while the number of subsets of $\Fn$ with cardinality $n$ is $\binom{q^n}{n}$.
\end{proof}

Now, let $\ell< n$  be two positive integers and let $0<s<n$ be another integer coprime with $n$. Let $(k_0,\ldots, k_{\ell-1})\in \Fn^{\ell}$ the secret key and
$\kappa(x)=k_0x+k_1x^{[s]}+\ldots+ k_{\ell-1}x^{[s(\ell-1)]}\in \mathcal L_n(\Fn)$ be the corresponding linearized polynomial. Consider a set of features $A=\left\{g_1,\ldots,g_n\right\}\subseteq \Fn$ given by $n$ $\Fq$-linearly independent elements.
Choose a random map $\lambda:\Fn\longrightarrow \Fn$ such that  $\lambda(x)\neq \kappa(x)$ for all $x\in B$, where $B= \Fn\smallsetminus A$.

Following the classical PFV scheme, we define the sets
\begin{align*}
\mathcal P_{auth} &  =  \left\{(x,\kappa(x)) \mid x \in A  \right\}, \\
\mathcal P_{chaff} & =  \left\{(x,\lambda(x)) \mid  x \in B \right\}, \\
\mathcal V & =  \mathcal P_{auth} \cup \mathcal P_{chaff}.
\end{align*}
$\mathcal P_{auth}$ is called \emph{set of authentic points}, $\mathcal P_{chaff}$ is the \emph{set of chaff points}, and $\mathcal V$  is called \emph{set of vault points}.

The last ingredients of the fuzzy vault scheme are the code
$$\C=\mathcal G_{\ell,s}(g_1,\ldots, g_n)$$
and an error correction decoding algorithm for $\C$.

For our constructions of the Linearized Polynomial Fuzzy Vault (LPFV), it is convenient to consider a Gabidulin code as a code whose codewords  consist of  evaluations of a linearized
polynomial $f\in \mathcal G_{\ell,s}$ over any set of $n$ linearly independent elements in $\Fn$. Concretely, we think of
a codeword as a  set of pairs $\left\{(g_i, y_i)\right\}_{i=1}^n$,  where $g_i \in \Fn$, are linearly independent over $\Fq$, and $y_i = f(g_i)$, for a linearized polynomial $f\in \mathcal G_{\ell,s}$.
In this framework, suppose that a witness attempts to gain access to the key, and submits a set  of features $W\subset \Fn$.

% As the error correction capability of $\C$ is $\lfloor(n-\ell)/2\rfloor$, the witness needs $|Z\cap \mathcal{P}_{auth}|\geq n-\lfloor(n-\ell)/2\rfloor = \lceil(n+\ell)/2\rceil$ to recover $\kappa(x)$ with the decoding algorithm.

Given Assumption \ref{as:lin}, if $Z\subseteq \mathcal{V}$ is the subset of vault points $(x,y)$ with $x\in W$, we can consider the $\Fq$-linear map 
$$ L_Z:\Fn\longrightarrow\Fn$$
such that $L_Z(x)=y$ for all $(x,y)\in Z$. Now, think of the received word $c'$ as consisting of the set of pairs  $\left\{(g_i, L_Z(g_i))\right\}_{i=1}^n$, for $g_i \in A$. The secret codeword of the LPFV scheme is instead $c$, given by the set of pairs $\left\{(g_i, \kappa(g_i))\right\}_{i=1}^n$. With this notation it is easy to see that
$$d_R(c,c')= \rk(\kappa-L_Z).$$

The following results relate the rank distance with the set difference, showing that the rank metric can be a good approximation of the set-difference metric. Let $d_{\Delta}(A,W):=|(A\backslash W) \cup (W\backslash A) |$ denote the set-difference between $A$ and $W$.
\begin{proposition}\label{pr:delta}
 %Under Assumption \ref{as:lin}, 
 In the setting of the LPFV scheme, suppose that the values  $\lambda(x)$, for $x \in B$ are chosen at random uniformly and independently in $\Fn \smallsetminus \left\{\kappa(x)\right\}$. Then 
\begin{enumerate}
 \item $2d_R(c,c') \leq d_{\Delta}(A,W)$.
\item Let $0\leq u\leq n$ be an integer. Then
$$\mathrm{Pr}\left\{2d_R(c,c') = d_{\Delta}(A,W) \mid |A\cap W|=u\right\}= \prod_{i=0}^{n-u-1}\frac{(q^n-q^i)}{(q^{n}-1)}=1+O(q^{-u-1}).$$
\end{enumerate}
\end{proposition}

\begin{proof}
\begin{enumerate}
\item Let $W$ be the set of features submitted, and let $u=|A\cap W|$. Then we have $d_{\Delta}(A,W)=2n-2|A\cap W|=2n-2u$. Consider now the $\Fq$-linear map $L_Z:\Fn\longrightarrow\Fn$ such that $L_Z(x)=y$ for $(x,y)\in Z$. The set of first coordinates of  $Z$ is an $\Fq$-basis of $\Fn$ and the linear map $\kappa-L_Z$ is $0$ on $A\cap W$. Therefore 
$$  d_R(c,c')=\rk(\kappa-L_Z)\leq n-u=\frac{d_{\Delta}(A,W)}{2}.$$
\item  Since the  $\lambda(x)$, for $x \in B$, are chosen at random uniformly and independently in $\Fn \smallsetminus \left\{\kappa(x)\right\}$, then the values $(L_Z-\kappa)(x)$, for $x\in W\smallsetminus (A\cap W)$ are chosen at random uniformly and independently in $\Fn\smallsetminus\{0\}$. Furthermore, the condition $2d_R(c,c') = d_{\Delta}(A,W)$ is equivalent to the condition that the values $(L_z-\kappa)(x)$, for $x\in W\smallsetminus (A\cap W)$ are linearly independent. Hence,
$$\mathrm{Pr}\left\{2d_R(c,c') = d_{\Delta}(A,W) \mid |A\cap W|=u\right\}= \frac{\left| \left\{A\in \Fq^{n\times (n-u)}\mid \rk(A)=n-u \right\}\right|}{(q^{n}-1)^{(n-u)}}.$$
\end{enumerate}

\end{proof}

%{\em Remark.} The probabilistic estimates above are linked to the probabilities of false negatives or positives in the authentication process.

\begin{theorem}\label{th:Set}
Under the same hypothesis of Proposition \ref{pr:delta}, the following statements hold.
\begin{enumerate}
 \item
 If  $d_{\Delta}(A,W)\leq 2\left\lfloor\frac{n-\ell}{2}\right\rfloor$, then the vault recovers the key $\kappa(x)$.
\item $$\mathrm{Pr}\left\{2d_R(c,c') = d_{\Delta}(A,W)\right\}=1+O(q^{-1}).$$
%If the vault recovers the key $\kappa(x)$, then 
%$$\mathrm{Pr}\left\{ d_{\Delta}(A,W)\leq 2\left\lfloor\frac{n-\ell}{2}\right\rfloor\mid \mbox{  the vault recovers } \kappa(x) \right\}\geq  ???$$
\end{enumerate}
\end{theorem}

\begin{proof}
\begin{enumerate}
\item  By Proposition \ref{pr:delta} we have  $2d_R(c,c') \leq d_{\Delta}(A,W)\leq  2\left\lfloor\frac{n-\ell}{2}\right\rfloor$. Therefore we are within the error-correction capability and we can correctly obtain the codeword $c$, and hence the key $\kappa(x)$.
\item We can write $\mathrm{Pr}\left\{2d_R(c,c') = d_{\Delta}(A,W)\right\}$ as
\begin{align*}
& \sum_{u=0}^n\mathrm{Pr}\left\{2d_R(c,c') = d_{\Delta}(A,W) \mid |A\cap W|=u\right\}\mathrm{Pr}\left\{|A\cap W|=u\right\} \\
                                         =& \sum_{u=0}^n \left( 1+O(q^{-u-1})\right)\mathrm{Pr}\left\{|A\cap W|=u\right\} \\
                                         = & \sum_{u=0}^n \left( 1+O(q^{-1})\right)\mathrm{Pr}\left\{|A\cap W|=u\right\} \\
                                         = & \left( 1+O(q^{-1})\right) \sum_{u=0}^n \mathrm{Pr}\left\{|A\cap W|=u\right\} \\
                                        = & 1+O(q^{-1}).
\end{align*}
\end{enumerate}
\end{proof}
{\em Remark.} Probabilistic results in Proposition \ref{pr:delta} and Theorem \ref{th:Set} do not depend on the probability distribution of the choice of the features. We are only assuming that our construction of the Linearized Polynomial Fuzzy Vault is made by choosing at random uniformly and independently the values  $\lambda(x)$ for $x\in B$.

\subsection{Generalization of the LPFV Scheme}

In our construction of the LPFV we  considered Gabidulin codes of length $n$ over $\Fn$. The motivation is that given a set of features $W$ satisfying Assumption \ref{as:lin}, the map $L_Z:\Fn \rightarrow \Fn$ is uniquely determined, and hence also the received word $c'$. 

We can generalize our LPFV  considering Gabidulin codes of length $n$ over  the field $\Fm$, where $n<m$, but we need to define the map $L_Z$ in a suitable way. 

Before explaining how to construct $L_Z$,  we can observe that an analogue of Lemma \ref{lem:lin} holds and it can be proved in the same way, but in this case the probability that %$W$ satisfies Assumption \ref{as:lin} is equal to 
the set of features is made of linearly independent elements is equal to

$$ \prod_{i=0}^{n-1}\frac{(q^m-q^i)}{(q^m-i)}=1+O(q^{-1-m+n}).$$

Now, let $\mathcal W$ and $\mathcal A$ be the $\Fq$-subspaces of $\Fm$ spanned respectively by $W$ and $A$.  First, we can observe that, in order to build the received word $c'$ as the set $\left\{(g_i, L_Z(g_i))\right\}_{i=1}^n$, we only need to define map $L_Z$ on $\mathcal A$. We propose the following construction. 

We first define the application $L_Z$ on $W$ as $L_Z(x)=y$ for all $(x,y)\in Z$. Then complete $W$ to a basis $B$ of $\mathcal A+\mathcal W$ , by adding the elements $g_i$ in increasing order with respect to the indices $i$. For those $g_i$, we set $L_Z(g_i)=\kappa(g_i)+\alpha^{q^i}$, where $\alpha \in \Fm$ and $\left\{\alpha^{q^i}\right\}_{i=0}^{m-1}$ is a normal basis of $\Fm$ as an $\Fq$-vector space.

In this way, our map is uniquely determined on $\mathcal A+\mathcal W$, and in particular on $\mathcal A$. Let again  $c$ be the codeword given by the set of pairs $\left\{(g_i, \kappa(g_i))\right\}_{i=1}^n$.  With this notation it is easy to see that
$$d_R(c,c')=\rk(\kappa-L_Z)_{|_{\mathcal A}}\leq \rk(\kappa-L_Z).$$
The following results are the analogues of Proposition \ref{pr:delta} and Theorem \ref{th:Set}, and they relate the rank distance of $c$ and $c'$ with the set difference of $A$ and $W$.

\begin{proposition}\label{pr:deltam}
 In the setting of the generalized LPFV scheme, suppose that the values  $\lambda(x)$, for $x \in B$ are chosen at random uniformly and independently in $\Fm \smallsetminus \left\{\kappa(x)\right\}$.  
\begin{enumerate}
\item Let the subspace distance be $d_S(\mathcal A,\mathcal W):=\dim_{\Fq}(\mathcal A)+\dim_{\Fq}(\mathcal W)-2\dim_{\Fq}(\mathcal A\cap \mathcal W)$. Then
 $$d_S(\mathcal A,\mathcal W) \leq 2d_R(c,c') \leq d_S(\mathcal A,\mathcal W)+2\rk(\kappa-L_Z)_{|_{\mathcal A\cap \mathcal W}}\leq d_{\Delta}(A,W).$$
\item Let $0\leq u \leq v \leq n$ be two integers. Then
 $$\pr\left\{2d_R(c,c') = d_{\Delta}(A,W) \mid |A\cap W|=u, \dim(\mathcal A \cap \mathcal W)=v \right\}=\prod_{i=n-v}^{n-u-1}\frac{(q^m-q^{i})}{q^m-1}.$$
\end{enumerate}
\end{proposition}

\begin{proof}
\begin{enumerate}
\item Following the construction of the map $L_Z$, we can write the subspace $\mathcal A$ as the direct sum of $\mathcal A \cap \mathcal W$ and the subspace $\widehat{\mathcal A}$, where $\widehat{\mathcal A}=\langle g_i \mid i \in I\rangle$ and $I\subset \{1,\ldots, n\}$ with $|I|=n-\dim_{\Fq}(\mathcal A \cap \mathcal W)$. Therefore we can write
\begin{equation}\label{eq:rk}
\rk(\kappa-L_Z)_{|_{\widehat{\mathcal A}}} \leq \rk(\kappa-L_Z)_{|_{\mathcal A}} \leq \rk(\kappa-L_Z)_{|_{\widehat{\mathcal A}}}+\rk(\kappa-L_Z)_{|_{\mathcal A\cap \mathcal W}}.
\end{equation}
Let $r=\dim_{\Fq}(\widehat{\mathcal A})$. By definition of the $L_Z$, we have  
$$ \rk(\kappa-L_Z)_{|_{\widehat{\mathcal A}}}=\rk(\alpha^{q^{i_1}},\ldots, \alpha^{q^{i_r}}).$$
By construction $\{\alpha^{q^i}\}_{i=0}^{m-1}$ is a normal basis of $\Fm$ over $\Fq$, and hence we can conclude that
$$ \rk(\kappa-L_Z)_{|_{\widehat{\mathcal A}}}=r=\dim_{\Fq}(\widehat{\mathcal A})=n-\dim_{\Fq}(\mathcal A\cap \mathcal W)=\frac{d_S(\mathcal A, \mathcal W)}{2}.$$
Substituting this equation in (\ref{eq:rk}) we obtain the first two inequalities.

For the last inequality we notice that the map $(\kappa-L_Z)_{|_{\mathcal A\cap \mathcal W}}$ is 0 on $|A\cap W|$, and therefore
$$\rk(\kappa-L_Z)_{|_{\mathcal A\cap \mathcal W}}\leq \dim_{\Fq}(\mathcal A\cap \mathcal W)-|A\cap W|.$$ 
Hence we can conclude that
$$ d_S(\mathcal A,\mathcal W)+2\rk(\kappa-L_Z)_{|_{\mathcal A\cap \mathcal W}}\leq 2n-2|A\cap W|=d_{\Delta}(A,W).$$
\item Let $u=|A\cap W|, v=\dim(\mathcal A \cap \mathcal W)$. Then we can write 
$$W=\left\{u_1,\ldots,u_{n-v},w_{n-v+1},\ldots,w_{n-u},g_{j_1},\ldots, g_{j_u} \right\},$$ 
where $u_i\notin \mathcal A$ for $i=1,\ldots,n-v$ and $w_i\in \mathcal A\smallsetminus A$ for $ i=n-v+1,\ldots, n-u$.
Therefore $2\rk(\kappa-L_Z)_{|_{\widehat{\mathcal A}}}=2n-2v$, and the condition 
$$\rk(\kappa-L_Z)_{|_{\mathcal A}}=\rk(\kappa-L_Z)_{|_{\widehat{\mathcal A}}}+\rk(\kappa-L_Z)_{|_{\mathcal A\cap \mathcal W}}=n-u$$
is equivalent to the condition 
$$\rk(\alpha^{q^{i_1}},\ldots, \alpha^{q^{i_{n-v}}}, (\kappa-L_Z)(w_{n-v+1}),\ldots,(\kappa-L_Z)(w_{n-u}))=n-u.$$
By hypothesis the values $(L_Z-\kappa)(w_i)$, for $i=n-v+1,\ldots, n-u$ are chosen at random uniformly and independently in $\Fm\smallsetminus\{0\}$, and we can conclude that the probability we are looking for is equal to
$$ \frac{\left| \left\{A\in \Fq^{m\times (v-u)} \mid \rk(A\mid X)=n-u\right\}\right|}{(q^m-1)^{(v-u)}}, $$
where $X$ is the matrix representation over $\Fq$ of the vector $(\alpha^{q^{i_1}},\ldots, \alpha^{q^{i_{n-v}}})$. Since 
$$ \left| \left\{A\in \Fq^{m\times (v-u)} \mid \rk(A\mid X)=n-u\right\}\right|=\prod_{i=n-v}^{n-u-1}(q^m-q^i) ,$$
this concludes the proof.
\end{enumerate}
\end{proof}

 \begin{theorem}\label{th:Setm}
Under the same hypothesis of Proposition \ref{pr:deltam}, the following statements hold.
\begin{enumerate}
 \item
 If  $d_{\Delta}(A,W)\leq 2\left\lfloor\frac{n-\ell}{2}\right\rfloor$, then the vault recovers the key $\kappa(x)$.
\item $$\mathrm{Pr}\left\{2d_R(c,c') = d_{\Delta}(A,W)\right\}=1+O(q^{-1-m+n}).$$
%If the vault recovers the key $\kappa(x)$, then 
%$$\mathrm{Pr}\left\{ d_{\Delta}(A,W)\leq 2\left\lfloor\frac{n-\ell}{2}\right\rfloor\mid \mbox{  the vault recovers } \kappa(x) \right\}\geq  ???$$
\end{enumerate}
\end{theorem}

\begin{proof}
\begin{enumerate}
\item The proof is essentially the same as the proof of Theorem \ref{th:Set}.1, using Proposition \ref{pr:deltam}.1.
\item In order to simplify the notation we introduce the events $D_u=\left\{ |A\cap W|=u\right\}$, $E_v=\left\{ \dim_{\Fq}(\mathcal A \cap \mathcal W)=v \right\}$ for $0\leq u, v \leq n$, and $X=\left\{2d_R(c,c') = d_{\Delta}(A,W)\right\}$. Then we have
\begin{align*}
\mathrm{Pr}\left\{X\right\}&= \sum_{0\leq u\leq v\leq n}\mathrm{Pr}\left\{X \mid D_u \cap E_v \right\}\mathrm{Pr}\left\{D_u \cap E_v\right\} \\
 &= \sum_{0\leq u\leq v\leq n}\left( 1+O(q^{-1-m-u+n})\right)\mathrm{Pr}\left\{D_u \cap E_v\right\} \\
                                          &= \sum_{0\leq u\leq v\leq n}\left( 1+O(q^{-1-m+n})\right)\mathrm{Pr}\left\{D_u \cap E_v\right\} \\
                                         &=\left( 1+O(q^{-1-m+n})\right) \sum_{0\leq u\leq v\leq n} \mathrm{Pr}\left\{D_u \cap E_v\right\} \\
                                        &=1+O(q^{-1-m+n}).
\end{align*}
\end{enumerate}
\end{proof}
{\em Remark.} Suppose one wants to use a generalized LPFV scheme with $n$ genuine features, and suppose moreover that a field $\Fq$ and an extension field $\Fm$, with $m\geq n$, are given. By Theorem \ref{th:Setm}.2 we can see that the bigger is $m$ the better is the approximation of the set difference with the rank distance. On the other hand, increasing $m$  implies an increase of the computational cost of the operations. Then one can  choose the best $m$ based on the application and the particular requirements of the context. 

\section{Applications}
\label{sec:5}

The schemes presented above can be applied in several scenarios for different purposes.
In this section we would like to give just a few examples.

One scenario for the fuzzy commitment scheme in the rank metric is the following. Suppose $B$ is the matrix used to create the stored tuple and imagine it as an image. It may happen for some reason that $B$ gets somehow damaged in a way that a few rows (or columns) are erased or anyway not the same as before. One can then authenticate with the new matrix $B'$ as long as not too many rows (or columns) are different. In another situations the matrix $B$ may be slightly changed into $B'$ by having all elements increased by a common error, and again the difference between the two matrices is a matrix of low rank, exactly $1$ in this case.

Another scenario involves a multi-factor authentication problem. Suppose that in order to perform authentication one needs a large number of conditions fulfilled, namely imagine a matrix with a large number of columns whereby condition number $i$ is fulfilled whenever column number $i$ equals a predetermined vector $v_i$. If you want to allow authentication as long as a certain big enough number of conditions are satisfied, then the fuzzy commitment scheme in the rank metric can be used.  Indeed having two matrices $A$ and $A'$ that both satisfy a certain condition corresponds to a zero column in the difference $A-A'$ which directly affects the rank distance between the two.

Applications for the linearized polynomial fuzzy vault scheme overlap with those of the standard fuzzy vault, i.e. we are considering authentication based on the set difference metric. It may be preferable to use the linearized version and decoding in the rank metric for certain choices and combinations of parameters which are usually dependent on the application. Also, the use of linear maps may be preferred for certain implementations. 

\section{Acknowledgments}
The authors were supported by Swiss National Science Foundation grant n.169510.

\bibliographystyle{abbrv}
%\bibliography{ref}
\bibliography{biblio2}

\begin{thebibliography}{10}

\bibitem{sc11}
M.~Baldi, M.~Bianchi, F.~Chiaraluce, J.~Rosenthal, and D.~Schipani.
\newblock On fuzzy syndrome hashing with {LDPC} coding.
\newblock In {\em 4th International Symposium on Applied Sciences in Biomedical
  and Communication Technologies (ISABEL)}, pages 1--5. ACM, 2011.

\bibitem{co16}
A.~Cossidente, G.~Marino, and F.~Pavese.
\newblock Non-linear maximum rank distance codes.
\newblock {\em Designs, Codes and Cryptography}, 79(3):597--609, 2016.

\bibitem{cr15}
J.~de~la Cruz, M.~Kiermaier, A.~Wassermann, and W.~Willems.
\newblock Algebraic structures of {MRD} codes.
\newblock {\em Advances in Mathematics of Communications}, (10):499--510, 2016.

\bibitem{sc12}
F.~Fontein, K.~Marshall, J.~Rosenthal, D.~Schipani, and A.-L. Trautmann.
\newblock On burst error correction and storage security of noisy data.
\newblock In {\em 20th International Symposium on Mathematical Theory of
  Networks and Systems (MTNS)}, pages 1--7, 2012.

\bibitem{ga85}
E.~M. Gabidulin.
\newblock Theory of codes with maximum rank distance.
\newblock {\em Problemy Peredachi Informatsii}, 21(1):3--16, 1985.

\bibitem{ho16}
A.~Horlemann{-}Trautmann and K.~Marshall.
\newblock New criteria for {MRD} and {G}abidulin codes and some rank-metric
  code constructions.
\newblock {\em arXiv:1507.08641, to appear in Advances in Mathematics of
  Communications}, 2016.

\bibitem{ju02}
A.~Juels and M.~Sudan.
\newblock A fuzzy vault scheme.
\newblock {\em Des. Codes Cryptography}, 38(2):237--257, 2006.

\bibitem{ju99}
A.~Juels and M.~Wattenberg.
\newblock A fuzzy commitment scheme.
\newblock In {\em 6th ACM conference on Computer and communications security},
  CCS '99, pages 28--36, 1999.

\bibitem{ks05}
A.~Kshevetskiy and E.~Gabidulin.
\newblock The new construction of rank codes.
\newblock In {\em International Symposium on Information Theory ({ISIT}),
  2005}, pages 2105--2108, 2005.

\bibitem{lo06}
P.~Loidreau.
\newblock A {W}elch--{B}erlekamp like algorithm for decoding {G}abidulin codes.
\newblock In {\em Coding and cryptography}, pages 36--45. Springer, 2006.

\bibitem{ma16}
K.~Marshall, D.~Schipani, A.-L. Trautmann, and J.~Rosenthal.
\newblock Subspace fuzzy vault.
\newblock In {\em Physical and Data-Link Security Techniques for Future
  Communication Systems}, pages 163--172. Springer, 2016.

\bibitem{ne16}
A.~Neri, A.-L. Horlemann-Trautmann, T.~Randrianarisoa, and J.~Rosenthal.
\newblock On the genericity of maximum rank distance and {G}abidulin codes.
\newblock {\em arXiv:1605.05972}, 2016.

\bibitem{ri04}
G.~Richter and S.~Plass.
\newblock Error and erasure decoding of rank-codes with a modified
  {B}erlekamp-{M}assey algorithm.
\newblock {\em ITG FACHBERICHT}, pages 203--210, 2004.

\bibitem{ro91}
R.~Roth.
\newblock Maximum-rank array codes and their application to crisscross error
  correction.
\newblock {\em IEEE Transactions on Information Theory}, 37(2):328--336, 1991.

\bibitem{sc10}
D.~Schipani and J.~Rosenthal.
\newblock Coding solutions for the secure biometric storage problem.
\newblock In {\em Information Theory Workshop (ITW), 2010}, pages 1--4, 2010.

\bibitem{sh15}
J.~Sheekey.
\newblock A new family of linear maximum rank distance codes.
\newblock {\em Advances in Mathematics of Communications}, (10):475--488, 2016.

\bibitem{si09}
D.~Silva and F.~Kschischang.
\newblock Fast encoding and decoding of {G}abidulin codes.
\newblock {\em International Symposium on Information Theory ({ISIT}), 2009},
  pages 2858--2862, 2009.

\end{thebibliography}

%\bibliography{ref}{}
%\bibliographystyle{spbasic}

\end{document}